\documentclass[aps,pra,twocolumn,superscriptaddress,amsmath,amssymb,nofootinbib]{revtex4}
\usepackage{amssymb}
\usepackage{multirow}
\usepackage{amsthm}
\def\beq{\begin{equation}}
\def\eeq{\end{equation}}
\usepackage{amsthm}
\usepackage{dcolumn}
\usepackage{bm}
\usepackage{graphicx}
\usepackage{tikz}
\usepackage[paperwidth=18.4cm, paperheight=26cm,left=2.5cm,right=2.0 cm,top=2.0cm,bottom=2.5cm,includehead,includefoot]{geometry}

\begin{document}
\newtheorem{corollary}{Corollary}
\newtheorem{definition}[corollary]{Definition}
\newtheorem{example}[corollary]{Example}
\newtheorem{lemma}[corollary]{Lemma}
\newtheorem{proposition}[corollary]{Proposition}
\newtheorem{theorem}[corollary]{Theorem}
\newtheorem{fact}[corollary]{Fact}
\newtheorem{property}[corollary]{Property}
\newtheorem{observation}[corollary]{Observation}
\newtheorem{question}[corollary]{Question}

\def\squareforqed{\hbox{\rlap{$\sqcap$}$\sqcup$}}
\def\qed{\ifmmode\squareforqed\else{\unskip\nobreak\hfil
\penalty50\hskip1em\null\nobreak\hfil\squareforqed
\parfillskip=0pt\finalhyphendemerits=0\endgraf}\fi}
\def\endenv{\ifmmode\;\else{\unskip\nobreak\hfil
\penalty50\hskip1em\null\nobreak\hfil\;
\parfillskip=0pt\finalhyphendemerits=0\endgraf}\fi}
\newenvironment{remark}{\noindent \textbf{{Remark~}}}{\qed}
\newcommand{\remarkTitle}[1]{\textbf{(#1)}}
\newcommand{\proofComment}[1]{\exampleTitle{#1}}
\newcommand{\bra}[1]{\langle #1|}
\newcommand{\ket}[1]{|#1\rangle}
\newcommand{\braket}[3]{\langle #1|#2|#3\rangle}
\newcommand{\ip}[2]{\langle #1|#2\rangle}
\newcommand{\op}[2]{|#1\rangle\!\langle #2|}
\newcommand{\tr}{{\operatorname{Tr}\,}}
\newcommand{\supp}{{\operatorname{supp}\,}}
\newcommand{\Sch}{{\operatorname{Sch}}}
\newcommand{\GHZ}{{\textrm{GHZ}}}
\newcommand{\slocc}{\stackrel{\textrm{\scriptsize SLOCC}}{\longrightarrow}}
\newcommand{\locc}{\stackrel{\textrm{\scriptsize LOCC}}{\longrightarrow}}
\newcommand{\rk}{{\operatorname{rk}}}
\newcommand{\sr}{{\operatorname{srk}}}
\newcommand{\pr}{{\operatorname{pr}}}
\newcommand{\E}{{\mathcal{E}}}
\newcommand{\F}{{\mathcal{F}}}
\newcommand{\h}{\mathcal{H}}
\newcommand{\diag} {{\rm diag}}
\newcommand{\nc}{\newcommand}
\nc{\ox}{\otimes}

\def\be{\begin{eqnarray}}
\def\ee{\end{eqnarray}}
\def\trace{\text{Tr}}
\def\lag{\langle}
\def\rag{\rangle}
\def\lt{\left}
\def\rt{\right}
\def\cf{\mathcal{F}}
\def\ch{\mathcal{H}}
\def\rmd{d}
\def\ci{\mathcal{I}}
\def\cc{\mathcal{C}}
\def\sig{\sigma}
\def\L{\Lambda}

%%%%%%%%%%%%%%%%%%%%%%%%%%%%%%%%%%%%%%%%%%%%%%%%%%%

\title{Invariant Perfect Tensors}

\author{Youning Li}%
\affiliation{Department of Physics, Tsinghua University, Beijing, People's Republic of China}%
\affiliation{Collaborative Innovation Center of Quantum Matter, Beijing 100190, People's Republic of China}%

\author{Muxin Han}%
\affiliation{Department of Physics, Florida Atlantic University, FL 33431, USA}%
\affiliation{Institut f\"ur Quantengravitation, Universit\"at Erlangen-N\"urnberg, Staudtstr. 7/B2, 91058 Erlangen, Germany}

\author{Markus Grassl}
\affiliation{Institut f\"ur Optik, Information und Photonik,
  Universit\"at Erlangen-N\"urnberg, 91058 Erlangen, Germany}%
\affiliation{Max-Planck-Institut f\"ur die Physik des Lichts,
  Leuchs Division, 91058 Erlangen, Germany}%

\author{Bei Zeng}
\affiliation{Department of Mathematics \& Statistics, University of
  Guelph, Guelph, Ontario, Canada}%
\affiliation{Institute for Quantum Computing, University of Waterloo,
  Waterloo, Ontario, Canada}%

\begin{abstract}
 Invariant tensors are states in the SU(2) tensor product
 representation that are invariant under the SU(2) action.  They play
 an important role in the study of loop quantum gravity. On the other
 hand, perfect tensors are highly entangled many-body quantum states
 with local density matrices maximally mixed.  Recently, the notion of
 perfect tensors recently has attracted a lot of attention in the
 fields of quantum information theory, condensed matter theory, and
 quantum gravity.  In this work, we introduce the concept of an
 invariant perfect tensor (IPT), which is a $n$-valent tensor that is
 both invariant and perfect.  We discuss the existence and
 construction of IPT.  For bivalent tensors, the invariant perfect
 tensor is the unique singlet state for each local dimension.  The
 trivalent invariant perfect tensor also exists and is uniquely given
 by Wigner's $3j$ symbol.  However, we show that, surprisingly, there
 does not exist four-valent invariant perfect tensors for any
 dimension.  On the contrary, when the dimension is large, almost all
 invariant tensors are perfect asymptotically, which is a consequence
 of the phenomenon of concentration of measure for multipartite
 quantum states.
\end{abstract}

\maketitle

\section{Introduction}
An invariant $n$-valent tensor $\psi$ is a state in the SU(2) tensor
product representation, and it is invariant under the SU(2) action.
Invariant tensors play a central role in the theory of Loop Quantum
Gravity (LQG) \cite{book,review,review1,rovelli2014covariant}, and
particularly the structure of Spin-Networks
\cite{Penrose,Rovelli:1995ac,Baez:1994hx}. The spin-network state, as
a quantum state of gravity, represents the quantization of geometry.
Classically an arbitrary three-dimensional geometry can be discretized
and built piece by piece by gluing polyhedral geometries.\footnote{The
  three-dimensional spatial geometry, quantized by spin-networks,
  serves as the initial data of four-dimensional gravity. } The
spin-network state quantizes the geometry made by polyhedra. As the
building block of spin-network, the $n$-valent invariant tensor
represents the quantum geometry of a polyhedron with $n$ faces
(explained in Appendix \ref{GeoInt}).

Briefly, the invariant tensor $\psi$ satisfies a quantum constraint
equation $\sum_{i=1}^n\mathbf{J}_i\psi=0$ where $\mathbf{J}_i$ denotes
the three SU(2) Lie algebra generators acting at the $i$-th tensor
component.  This equation is a quantum analog of the (flat) polyhedron
closure condition $\sum_{i=1}^n\vec{A}_i=0$ in three-dimensional
space,\footnote{The closure equation generalized to constant curvature
  polyhedron has been proposed in
  \cite{curvedMink,HHKR,Charles:2015lva}.} where each $\vec{A}_i$ is
the oriented area vector of the $i$-th polyhedron face. The one-to-one
correspondence between a flat geometrical polyhedron and $n$ vectors
$\vec{A}_1,\cdots,\vec{A}_n$ satisfying the closure condition is known
as the Minkowski Theorem \cite{Minkowski}.  Comparing the classical
and the quantum closure equations identifies $\mathbf{J}_i$ to be the
quantization of $\vec{A}_i$. The invariant tensor $\psi$ thus
represents a quantized geometrical polyhedron. In particular, each
polyhedron face area is quantized by
$(\vec{A}_i\cdot\vec{A}_i)^{1/2}\sim
(\mathbf{J}_i\cdot\mathbf{J}_i)^{1/2}=\sqrt{j_i(j_i+1)}$, which is the
famous quantum area spectrum in LQG \cite{Rovelli1995,ALarea}. The
invariant tensor $\psi$ with fixed $j_i$ for each component is a
quantum parameterization of the shapes of polyhedra with fixed face
areas \cite{shape}. The three-dimensional quantum geometry is
constructed by collecting a large number of invariant tensors
representing different quantum geometrical polyhedra. It corresponds
to the kinematics of four-dimensional quantum gravity.

On the other hand, another notion of special tensors, known as
\emph{perfect tensors}, has recently attracted a lot of attention from
researchers in quantum information theory, condensed matter theory,
and quantum
gravity~\cite{hosur2015chaos,Pastawski:2015qua,Almheiri:2014lwa}. In
this work, we consider $n$-valent tensors with dimension $d$ for each
component (i.e., $n$-qudit quantum states).  In this setting, the
perfect tensor is a highly entangled many-body quantum state, where
its reduced density matrix of any part of the system, involving up to
half of the total number of particles of the system, is maximally
mixed.

In terms of quantum error-correcting codes, a perfect tensor is a code
with large code distance that is half of the system size.  Intimate
connections between quantum error-correcting codes, perfect tensors,
information scrambling in chaotic many-body quantum systems, and
systems with holographic duals, have recently been revealed.

Perfect tensors have been employed to construct the \emph{Tensor
  Network} as a Conformal Field Theory (CFT) ground state, which
realizes the AdS/CFT correspondence
\cite{Pastawski:2015qua,Bhattacharyya:2016hbx}.  In particular, the
perfect tensor network provides an interesting illustration of how the
Ryu-Takayanagi formula of Holographic Entanglement Entropy
(HEE)\footnote{In the context of bulk-boundary duality, Ryu-Takayanagi
  Formula conjectures that the CFT entanglement entropy of a spatial
  region $A$ is proportional the minimal area of the bulk codim-2
  surface attached to $\partial A$ \cite{Ryu:2006bv}.} emerges from
many body quantum system.

Furthermore, recently it has been shown that perfect tensors represent
quantum channels which are of strongest quantum chaos
\cite{hosur2015chaos}.  The quantum transition defined by perfect
tensors turns out to maximally scramble the quantum information such
that the initial state cannot be recovered by local measurements. In
\cite{hosur2015chaos} it was also suggested that a perfect tensor
should represent the holographic quantum system dual to the bulk
quantum gravity with a black hole.\footnote{The recent AdS/CFT
  computation reveals that a black hole should be dual to a quantum
  system of fastest scrambling \cite{Maldacena:2015waa}, which is
  consistent with the scrambling feature of perfect tensors.}

Given that invariant tensors and perfect tensors relate to quantum
gravity from different perspectives, it is then highly desired to
incorporate the idea of perfect tensors with that of invariant
tensors, a new concept that we call it Invariant Perfect Tensor (IPT).
This work is also motivated by the recent result in \cite{HanHung}, in
which the tensor network and HEE Ryu-Takayanagi formula emerge from
LQG spin-network with invariant tensors.

The existence and construction of invariant tensors or perfect tensors
are to some extent well understood.  Moreover, among bivalent tensors,
i.e., bipartite quantum states, the existence of IPT is also
understood, which is nothing but the spin singlet state.  We will show
that, among trivalent tensors, invariant perfect tensors can also be
constructed uniquely from Wigner's $3j$ symbol.  However, the existence
and constructions of IPT have been unknown for $n>3$.  As a surprising
result, we show that there does not exist any IPT for $n=4$, for any
local dimension $d$.  On the other hand, however, a random $4$-valent
invariant tensor is nearly perfect for large $d$.  In other words,
random invariant tensors also demonstrate a similar behavior of
concentration of measure of generic quantum states, although the
entropy convergence rate to the maximum possible value turns out to be
slower.  Our method and results also shed light on more general
structure of IPT of $n>4$.

We organize our paper as follows: in Section~\ref{sec:pre}, we introduce
basic notations and preliminaries on SU(2) representations.  In
Section~\ref{sec:tri}, we discuss the construction of $3$-valent IPT
using Wigner's $3j$ symbol. In Section~\ref{sec:nogo}, we prove a no-go
theorem that there does not exist $4$-valent IPT tensor.  In
Section~\ref{sec:ram}, we discuss random $4$-valent invariant tensors and
show that they are nearly perfect in the large dimension $d$
limit. Finally, a brief discussion will be given in
Section~\ref{sec:dis}.

\section{Notations and Preliminaries}
\label{sec:pre}
A multipartite quantum system of $n$-particles has a Hilbert space
$\mathcal{H}_n=\otimes_{i=1}^n V_{j_i}$, where each $V_{j_i}$ is a
spin-${j_i}$ with dimension $d_i=2j_i+1$. The spin angular momentum
operators have commutation relations given by
$[J^a,J^b]=i\epsilon^{abc}J^c$. An $n$-valent tensor is a vector
$\ket{\psi_n}$ in $\mathcal{H}_n$.

Let the total spin operator be
\begin{equation}
\mathbf{J}=\sum_{i=1}^n \mathbf{J}^i.
\end{equation}
An $n$-valent tensor $\ket{\psi_n}$ is invariant if it satisfies
\begin{equation}
\mathbf{J}\ket{\psi_n}=0.\label{invtensor}
\end{equation}

In the tensor product $\otimes_{i=1}^n V_{j_i}$ of $n$ $SU(2)$
irreducible representations $V_{j_i}$ labeled by spins
$j_1,\cdots,j_n$, the dimension of the subspace
$\mathrm{Inv}(\otimes_{i=1}^n V_{j_i})$ spanned by the invariant
states is given by the following formula \cite{Verlinde:1988sn}
\begin{alignat*}{5}
&\dim\left[\mathrm{Inv}(\otimes_{i=1}^n V_{j_i})\right]\nonumber\\
&\quad=\frac{2}{\pi}\int_0^{\pi} d \theta\, \sin^2(\theta/2)\prod_{i=1}^n\frac{\sin((j_i+\frac{1}{2})\theta)}{\sin(\theta/2)}.
\end{alignat*}
For invariant $n$-qudit states, take $j_1=\cdots=j_n=j$ with $d=2j+1$.

For adding angular momentums, we use the standard Clebsch-Gordan
coefficients (CGCs) that are written as
\begin{equation}
C^{\, j_1\,\, j_2}_{m_1\, m_2\, J\, M}=\langle j_1m_1;j_2m_2\ket{JM}.
\label{3j-def}
\end{equation}
We also use Wigner's $3j$ symbol that is given in terms of CGCs as
\begin{alignat}{5}
&\begin{pmatrix}
    j_1 & j_2 & j_3 \\
    m_1 & m_2 & m_3
  \end{pmatrix}\nonumber\\
&\qquad=\frac{(-1)^{j_1-j_2-m_3}}{\sqrt{2j_3+1}}\langle j_1m_1;j_2m_2\ket{j_3m_3}.
\end{alignat}

Note that in order the $3j$ symbol is nonzero, the spins $j_1,j_2,j_3$
have to satisfy the triangle inequality:
\[
|j_1-j_2|\leq j_3\leq j_1+j_2.
\]
It leads to the geometrical interpretation of $3j$ symbol as a triangle
in two-dimensional Euclidean space, whose three edge lengths are $j_1,j_2,j_3$.

The $3j$ symbol can be chosen to be purely real, and it is invariant
under an even permutation of its columns:
\begin{alignat}{5}
\begin{pmatrix} j_1 &
  j_2 & j_3\\ m_1 & m_2 & m_3
\end{pmatrix}
&{}=
\begin{pmatrix}
  j_2 & j_3 & j_1\\
  m_2 & m_3 & m_1
\end{pmatrix}\nonumber\\
&{}=
\begin{pmatrix}
  j_3 & j_1 & j_2\\
  m_3 & m_1 & m_2
\end{pmatrix}.
\end{alignat}
Moreover, the $3j$ symbol obeys the following orthogonality relation
\begin{alignat}{5}
&\sum_{m_1 m_2}
\begin{pmatrix}
  j_1 & j_2 & j\\
  m_1 & m_2 & m
\end{pmatrix}
\begin{pmatrix}
  j_1 & j_2 & j'\\
  m_1 & m_2 & m'
\end{pmatrix}\nonumber\\
&\qquad{}=\frac{1}{(2j+1)}\delta_{j j'}\delta_{m m'}.
\end{alignat}

An $n$-qudit invariant state is an $n$-valent tensor with
$d_i=2j_i+1=d$. An $n$-qudit state/tensor $\ket{\psi_n}$ is perfect if
for any bipartition, whose number of particles $k$ in the smaller part
satisfies $1\leq k\leq \lfloor n/2\rfloor$, the entropy of the reduced
density matrix is maximal.  An $n$-qudit state $\ket{\psi_n}$ is an
invariant perfect tensor (IPT) if it is both invariant and
perfect. Our goal is to study the existence and construction of IPTs
for $n$-qudit states.

\section{\boldmath Three-valent IPT: Wigner's $3j$ Symbols}
\label{sec:tri}

We consider three-valent IPTs in this section, and we find that for
$n=3$ there is a unique invariant tensor of SU(2) (up to a rescaling),
which is also a perfect tensor.

Consider a tensor product $\otimes_{i=1}^3 V_{j_i}$ of three SU(2)
irreducible representations labeled by spins $j_1,j_2,j_3$. It is
well-known that the subspace $\mathrm{Inv}(\otimes_{i=1}^3 V_{j_i})$
of invariant tensors is one-dimensional in the case of rank three.
The normalized invariant tensor is given by Wigner's $3j$ symbol:
\[
(\psi_3)^{j_1,j_2,j_3}_{m_1,m_2,m_3}=\begin{pmatrix}
    j_1 & j_2 & j_3 \\
    m_1 & m_2 & m_3
  \end{pmatrix},
\]
where the indices $m_1,m_2,m_3$ transform under SU(2) in $j_1,j_2,j_3$
representations, respectively.

Now consider the state $\ket{\psi_3}$ given by
\begin{equation}
\label{eq:tri}
\sum_{m_1,m_2,m_3}\kern-3mm(\psi_3)^{j_1,j_2,j_3}_{m_1,m_2,m_3} |j_1 m_1 \rangle |j_2 m_2\rangle |j_3 m_3\rangle.
\end{equation}
To get an invariant tensor satisfying $\mathbf{J}|\psi_3\rangle =0$,
we need to have the $3j$ symbols given by the coefficients with which
three angular momenta added to zero.  We now show that in this case,
$\ket{\psi_3}$ is also perfect.

For any choice of two spins $j_1,j_2$ out of $j_1,j_2,j_3$, we define
the reduced density matrix
$\rho_3=\tr_{j_1,j_2}|\psi_3\rangle\langle\psi_3|$. The orthogonality
relation implies that
\begin{equation}
\langle j_3,m_3|\rho_3|j_3m_3'\rangle=\frac{1}{(2j_3+1)}\delta_{m_3,m_3'},\label{cong}
\end{equation}
and hence the entanglement entropy is maximal $S_3=\ln (2j_3+1)$.  So
we proved that $\psi_3$ constructed from Wigner's $3j$ symbols is an
invariant perfect tensor.

For a three-qudit state, we have $j_1=j_2=j_3$, and $d=2j_1+1$. Notice
that for even $d$, $(\psi_3)^{j_1,j_2,j_3}_{m_1,m_2,m_3}$ is always
zero, and for odd $d$, the invariant tensor is unique.

As the simplest example, we take $j_1=j_2=j_3=1$. The $3j$ symbol
simply give the $\epsilon$ symbol (anti-symmetric tensor)
\begin{eqnarray}
(\psi_3)^{1,1,1}_{m_1,m_2,m_3}&=&
\begin{pmatrix}
  1 & 1 & 1\\
  m_1 & m_2 & m_3
\end{pmatrix}\nonumber\\
&=&\frac{1}{\sqrt{6}}\epsilon_{m_1,m_2,m_3}.
\end{eqnarray}
It is readily checked that the reduced density matrix of any single
particle is maximally mixed.

\section{Four-valent IPT: a no-go theorem}
\label{sec:nogo}

In this section, we discuss the existence of $4$-valent IPT. In this
case, the dimension of the invariant subspace equals the qudit
dimension $d$.  On the other hand, the perfect tensors exist for any
$d>2$, possibly with the exception of $d=6$ \cite{PhysRevA.92.032316}.
Since the dimension of invariant subspace grows linearly with $d$, one
might expect that it should be possible to at least find an IPT when $d$
is large.  However, surprisingly, it turns out that there does not
exist IPT for any $d$.

\begin{theorem}
\label{th:main}
There does not exist $4$-valent IPTs, for any $d$.
\end{theorem}

To prove this theorem, we start from writing down a general form of
invariant tensors.  For the $4$-qudit invariant tensor, by choosing a
coupling scheme, we can formulate the state in terms of the
Clebsch-Gordan coefficients as follows:
\begin{widetext}
\begin{equation}
|\psi_4\rangle=\sum_{J=0}^{2j}\alpha(J)\sum_{\text{all }m_is,M}C^{\, j\,\, j}_{m_1\, m_2\, J\, M}C^{\, j\,\, j}_{m_3\, m_4\, J\, -M}
C^{\, J\,\, J}_{M\, -M\, 0\, 0}\ket{m_1,m_2,m_3,m_4},\label{4-pure}
\end{equation}
\end{widetext}
where all $m_is$ run from $-j$ to $j$ and $M$ runs from $-J$ to $J$.

The perfect state condition requires that
\begin{equation}
\rho_{34}=\rho_{24}=\rho_{23}=\frac{1}{d^2}\mathbb{I}_{d^2},
\label{perf}
\end{equation}
where $\rho_{ij}=\tr_{\overline{ij}}\ket{\psi}\bra{\psi}$ is the reduced
density matrix of the $i,j$ particles, and $\mathbb{I}_{d^2}$ is the
identity matrix of size $d^2\times d^2$, where $d=2j+1$.

We will show that Eq.~\eqref{perf} cannot be satisfied for any $d$.
We first examine the consequence of
$\rho_{34}=\frac{1}{d^2}\mathbb{I}_{d^2}$, which is given by the
following lemma.

\begin{lemma}\label{lemma2}
If $\rho_{34}=\frac{1}{d^2}\mathbb{I}_{d^2}$,
then \[|\alpha(J)|=\frac{\sqrt{2J+1}}{2j+1}.\]
\end{lemma}

\begin{proof}
According to Eq.~\eqref{4-pure}, the matrix element of $\rho_{34}$
labeled by ${m_3m_4,m_3'm_4'}$ is given by
\begin{eqnarray}
\sum_{J,M}\frac{|\alpha(J)|^2}{2J+1}
C^{\, j\,\, j}_{m_3\, m_4\, J\, -M}C^{\, j\,\, j}_{m_3'\, m_4'\, J\, -M}.\label{tr12}
\end{eqnarray}
Substituting the definition of CGCs in Eq.~\eqref{3j-def} into
Eq.~\eqref{tr12}, we get
\begin{alignat*}{5}
&\sum_{J,M}\frac{|\alpha(J)|^2}{2J+1}
C^{\, j\,\, j}_{m_3\, m_4\, J\, -M}C^{\, j\,\, j}_{m_3'\, m_4'\, J\, -M}\nonumber\\
&\quad=\bra{jm_3;jm_4}\hat{O}\ket{jm_3';jm_4'}\nonumber\\
&\quad=\frac{1}{d^2}\delta_{m_3m'_3}\delta_{m_4m'_4},
\end{alignat*}
where
$\hat{O}=\sum_{J,M}\frac{|\alpha(J)|^2}{2J+1}\ket{J,-M}\bra{J,-M}$. The
second equality is true for any element, which means that
\begin{equation}
\sum_{J,M}\frac{|\alpha(J)|^2}{2J+1}\ket{J,-M}\bra{J,-M}=\frac{1}{d^2}\mathbb{I}_{d^2}.
\end{equation}

The completeness of the basis $\{\ket{J,M}\}$ implies that the
identity operator has the unique decomposition as
\begin{alignat*}{5}
\mathbb{I}_{d^2}=\sum_{J,M}\ket{J,-M}\bra{J,-M},
\end{alignat*}
so we conclude that
\begin{alignat*}{5}[
|\alpha(J)|=\frac{\sqrt{2J+1}}{2j+1}.
\end{alignat*}
\end{proof}

By Lemma~\ref{lemma2}, we can rewrite
\begin{alignat*}{5}
\alpha(J)=\frac{\sqrt{2J+1}}{2j+1}\omega(J),
\end{alignat*}
where $\omega(J)$ is a phase factor.

Now we further examine the consequence of
$\rho_{24}=\rho_{23}=\frac{1}{d^2}\mathbb{I}_{d^2}$, and show that no choice of
$\omega(J)$ can satisfy both conditions.  The key idea is the
following: $\rho_{24}$ is obtained by tracing out the particles $1,3$
from $\ket{\psi_4}$; on the other hand, $\rho_{23}$ can be obtained by
first swapping particles $3$ and $4$ in $\ket{\psi_4}$, then tracing
out the particles $1,3$.  Due to the form of $\ket{\psi_4}$ that
involves Clebsch-Gordan coefficients, the permutation will result in various
$(-1)^J$ factors.  Consequently, we will end up with two equations for
$\omega(J)$ that contradict each other, which then proves that no
choice of $\omega(J)$ can lead to
$\rho_{24}=\rho_{23}=\frac{1}{d^2}\mathbb{I}_{d^2}$.

To be more concrete, $\rho_{24}=\frac{1}{d^2}\mathbb{I}_{d^2}$
leads to the equation
\begin{equation}\label{eq:omega1}
\sum_{J}(-1)^J\omega(J)C^{\, j\,\, j}_{-j\, j\, J\, 0}C^{\, j\,\, j}_{-j\, j\, J\, 0}=e^{i\theta},
\end{equation}
and $\rho_{24}=\frac{1}{d^2}\mathbb{I}_{d^2}$
leads to an equation
\begin{equation}
\label{eq:omega2}
\sum_{J}\omega(J)C^{\, j\,\, j}_{-j\, j\, J\, 0}C^{\, j\,\, j}_{-j\,
  j\, J\, 0}=e^{i\theta'}.
\end{equation}
(See Appendix~\ref{app:eq} for details concerning the derivation of
Eqs.~\eqref{eq:omega1} and~\eqref{eq:omega2}.)  Now we show that these two
equations cannot be satisfied simultaneously, which is given by the
following lemma.

\begin{lemma}\label{lemma3}
$\omega(J)=1$ is the only solution of $\omega(J)$ to the equation
\begin{equation}
\sum_{J}\omega(J)C^{\, j\,\, j}_{-j\, j\, J\, 0}C^{\, j\,\, j}_{-j\, j\, J\, 0}=1,\label{eq:lemma2}
\end{equation}
when $\omega(J)$ is a phase factor.
\end{lemma}

\begin{proof}
Firstly, it can be easily checked that, $\omega(J)=1$ is a
solution. Suppose we have another phase factor $\omega_1(J)$, which
satisfies Eq.~\eqref{eq:lemma2}, so
\[
\sum_{J}[1-\operatorname{Re}(\omega_1(J))]C^{\, j\,\, j}_{-j\, j\, J\, 0}C^{\, j\,\, j}_{-j\, j\, J\, 0} =0,
\]
however, $C^{\, j\,\, j}_{-j\, j\, J\, 0}C^{\, j\,\, j}_{-j\, j\, J\,
  0} > 0$, and $\operatorname{Re}(\omega_1(J))\leq 1$, for all $J ={0,1\ldots,2j}$.
This directly leads to the fact that $\omega_1(J)=1$.
\end{proof}

Using Lemma~\ref{lemma3}, one sees the intrinsic contradiction between
Eq.~\eqref{eq:omega1} and Eq.~\eqref{eq:omega2}, so the conditions
$\rho_{34}=\rho_{24}=\rho_{23}=\frac{1}{d^2}\mathbb{I}_{d^2}$ cannot
be satisfied simultaneously for any $d$. One may easily verify that
$\alpha(J)=(-1)^J\frac{\sqrt{2J+1}}{2j+1}$ does satisfy
$\rho_{34}=\rho_{24}=\frac{1}{d^2}\mathbb{I}_{d^2}$, and is also the
unique solution after neglecting an unimportant global phase. In other
words, if $\rho_{34}=\rho_{24}=\frac{1}{d^2}\mathbb{I}_{d^2}$ for any
$4$-valent invariant tensor $\ket{\psi_4}$, then we cannot have
$\rho_{23}=\frac{1}{d^2}\mathbb{I}_{d^2}$ at the same time. This hence
proves Theorem~\ref{th:main}.

\section{Random Invariant Tensor and Asymptotical Perfectness}
\label{sec:ram}

In the last section, we have shown that there does not exist any
$4$-valent IPT. Then a natural question is whether there exists an
invariant tensor that is `nearly perfect'.  To examine this question,
we would like to consider the limit $j\rightarrow\infty$ ($d=2j+1$). We
know that, in this case, a random tensor exhibits the phenomenon of
`concentration of measure', where for any bipartition, the
entanglement entropy of the reduced state is near the maximally
possible, asymptotically as $j\rightarrow\infty$. Now the question
becomes whether this `concentration of measure' phenomenon will also
show up in the space of invariant tensors. We give an affirmative
answer in this section for the case of $4$-valent invariant tensor.

Given an invariant tensor $|I\rangle \in
\mathrm{Inv}_{SU(2)}(V_{j_1}\otimes\cdots\otimes V_{j_4})$, we define
the density matrix $\rho=|I\rangle\langle I|$. We consider an
arbitrary bipartition into two pairs.  Without loss of
generality, we consider the reduced density matrix
$\rho_{34}=\trace_{12}\rho$ by tracing out the degrees of freedom in
$V_{j_1}\otimes V_{j_2}$.  The second Renyi entropy $S_2$ of
$\rho_{34}$ is given by
\begin{equation}
  e^{-S_2}=\frac{\trace
    \rho_{34}^2}{(\trace \rho_{34})^2}.\label{S2}
\end{equation}
It is not hard to check that the numerator
\begin{equation}
  Z_1\equiv \trace
  \rho_{34}^2=\trace\lt[\lt(\rho\otimes\rho\rt)\cf_{34}\rt],
\end{equation}
where the last trace is over the space $(V_{j_1}\otimes\cdots\otimes
V_{j_4})^{\otimes 2}$. $\cf_{34}$ is a swap operator that swaps
particles $3$ and $4$.  The denominator of Eq.~\eqref{S2} can be
written similarly as
\begin{equation}
  Z_0\equiv(\trace
  \rho_{34})^2=\trace\lt[\rho\otimes\rho\rt].
\end{equation}
We randomly sample the invariant tensors $|I\rangle$ in the invariant
subspace $\ch_{\text{inv}}=\mathrm{Inv}_{SU(2)}(V_{j_1}\otimes\cdots\otimes
V_{j_4})$, and consider the average
\begin{equation}
\overline{Z_1}=\trace\lt[\overline{\lt(\rho\otimes\rho\rt)}\cf_{34}\rt].
\end{equation}
Direct calculation by using Schur's Lemma of Haar random average
\cite{church} shows that (see Appendix~\ref{app:Z1Z0})
\begin{equation}
\overline{Z_1}=\frac{2\sum_{I}(2I+1)^{-1}}{\dim(\ch_{\text{inv}})^2+\dim(\ch_{\text{inv}})}
\end{equation}
and $\bar{Z}_0=1$.  Therefore, the averaged second Renyi entropy is given by
\begin{alignat}{5}
\overline{S_2}&{}=-\ln\frac{\overline{Z_1}}{\overline{Z_0}}\nonumber\\
&{}=\ln \lt[\dim\lt(\ch_{\text{inv}}\rt)^2+\dim\lt(\ch_{\text{inv}}\rt)\rt]\nonumber\\
&\quad{}-\ln\lt(2\sum_{I}(2I+1)^{-1}\rt).\label{ZZ}
\end{alignat}

When all spins $j_1=j_2=j_3=j_4=j$ are equal, we have
$\dim\lt(\ch_{\text{inv}}\rt)=2j+1$ and
\begin{alignat}{5}
\overline{S_2}&=\ln \lt[(2j+1)^2+(2j+1)\rt]\nonumber\\
&\quad{}-\ln\lt(2\sum_{I=0}^{2j}(2I+1)^{-1}\rt).
\end{alignat}
Let $j\to\infty$ asymptotically, the leading behavior of $\overline{S_2}$ is
\begin{equation}
\overline{S_2}\sim \ln \lt[(2j+1)^2\rt].
\end{equation}
Although $\ln\lt(2\sum_{I=0}^{2j}(2I+1)^{-1}\rt)$ is also divergent as $j\to\infty$, the divergence is much slower than $\ln \lt[(2j+1)^2\rt]$, indeed,
\begin{equation}
\lim_{j\to\infty} \frac{\ln\lt(2\sum_{I=0}^{2j}(2I+1)^{-1}\rt)}{\ln \lt[(2j+1)^2\rt]}=0.
\end{equation}

We also estimate the fluctuation of $\bar{S_2}$. In fact, for any
small $\delta>0$, there is a large probability
\begin{equation}
  P_\delta=1-\frac{3\pi^2}{\delta^2\L^2},
\end{equation}
which is close to $1$ as $j\to\infty$ (since
$\L=\sum_{I=0}^{2j}\frac{1}{2I+1}\to\infty$), such that
$\lt|S_2-\overline{S_2}\rt|\leq\delta$, i.e., the second Renyi entropy
is close to the average value
$\overline{S_2}\sim\ln\lt[(2j+1)^2\rt]$. The derivation of the above
result is presented in Appendix~\ref{sec:bound} \footnote{The idea of
  the proof is similar to \cite{Qi1}}).

Because the von Neumann entropy is lower bounded by the second Renyi,
i.e., $S\geq S_2$, and $\ln\lt[(2j+1)^2\rt]$ is the maximal value of
the entanglement entropy of the 4-valent tensor state, we have for the
Von Neumann entropy
\begin{equation}
  S\sim\ln\lt[(2j+1)^2\rt].
\end{equation}
The state is maximally entangled for any partition into two pairs,
asymptotically as $j\to\infty$.  Therefore, the random invariant tensor
is asymptotically a perfect tensor.

\section{Discussion}
\label{sec:dis}
We have introduced the concept of Invariant Perfect Tensors (IPT) and
discussed their existence and construction.  For $3$-valent tensor, IPT
exist for integer spin $j$ and is given by the unique spin zero state
whose coefficient is Wigner's $3j$ symbol. We  showed that there does not
exist $4$-valent IPT for any single particle spin $j$. On the other
hand, a random $4$-valent invariant tensor is asymptotically perfect.

It is natural to ask about the case of $n>4$.  It turns out that the
situation is more complicated, and the method used to prove
Theorem~\ref{th:main} does not directly apply for $n>4$. However, one
may expect the permutation of particles in the invariant subspace may
still cause certain contradictions such that some of the reduced
density matrices cannot simultaneously be identity.  Numerical results
for small local dimensions and $n=5$ and $n=6$ indicate such
contradictions.  One may guess that there might be some fundamental
structural reason that IPT might not exist for $n>4$, although IPT
might appear asymptotically when $j$ is large. We leave this for
future research.

\section*{Acknowledgements}

We thank Jianxin Chen, Shawn Cui, Cheng Guo, Guilu Long, Dong Ruan,
Nengkun Yu, Ling-Yan Hung and Yidun Wan for helpful discussions. MH acknowledges Yidun Wan and Ling-Yan Hung at Fudan University, Wei Song at Yau's Institute of Mathematics (Tsinghua University), for their hospitality during his visit. MH acknowledges support from the
US National Science Foundation through grant PHY-1602867, and the
Start-up Grant at Florida Atlantic University, USA. BZ is supported by
NSERC and CIFAR.

\appendix

\section{Geometrical interpretation of invariant tensor}\label{GeoInt}

The origin of geometrical interpretation traces back to a classic
theorem by Minkowski, which states the following: Given a set of
vectors $\vec{A}_1,\cdots,\vec{A}_n\in\mathbb{R}^3$ satisfying a
closure condition $\sum_{i=1}^n\vec{A}_i=0$, then there is a unique
polyhedron in $\mathbb{R}^3$ with $n$ faces, whose face areas is given
by $|\vec{A}_i|$ and the normal of each face is given by
$\vec{A}_i/|\vec{A}_i|$.  Therefore, a classical polyhedron geometry can
be parameterized by the oriented face area vectors
$\vec{A}_1,\cdots,\vec{A}_n$ subject to the closure condition.

Loop Quantum Gravity (LQG) provides the result that the polyhedron
geometry can be quantized.  The quantum polyhedron geometry is obtained
by promoting the vectors $\vec{A}_1,\cdots,\vec{A}_n$ to vector-valued
operators $\hat{\vec{A}}_1,\cdots,\hat{\vec{A}}_n$. LQG derives the
commutation relation between the operators
\cite{Ashtekar:1998ak,Thiemann:2000bv}:
\begin{equation}
\left[\hat{A}_i^a,\hat{A}_j^b\right]=8\pi i\gamma \ell_P^2 \delta_{ij}\epsilon^{abc}\hat{A}_i^c,\label{Acomm}
\end{equation}
where $a,b,c=1,2,3$ are the indices of vector components,
$i,j=1,\cdots,n$ label the faces of polyhedron; $\ell_P=G_N\hbar$ is
the Planck length; $\gamma$ is called the Barbero-Immirzi parameter in
LQG.  It is easy to see that different faces correspond to different
degrees of freedom, which are commutative.  For a given face $i$, the
vector components of $\hat{\vec{A}}_i$ are non-commutative. The
commutation relation is the same as the commutation relation of
angular momentum operator $[J^a,J^b]=i\epsilon^{abc}J^c$, or
equivalently the commutation relation of the Lie algebra
$\mathfrak{su}(2)$.

The Hilbert space may be chosen as a tensor product of SU(2) irreps
$\otimes_{i=1}^n V_{j_i}$, to represent the above operator algebra
Eq.~\eqref{Acomm}. Each $\hat{A}^a_i$ ($a=1,2,3$) is represented as
the $\mathfrak{su}(2)$ generator $J^a$ acting on the $i$-th copy of
irrep $V_{j_i}$
\[
\hat{A}^a_i=8\pi \gamma \ell_P^2 J^a.
\]
However, recall that classically there is the closure condition
constraining the data $\vec{A}_1,\cdots,\vec{A}_n\in\mathbb{R}^3$. The
closure condition has to be promoted to an operator constraint
\begin{equation}
\sum_{i=1}^n\hat{\vec{A}}_i\psi=8\pi \gamma \ell_P^2\sum_{i=1}^n\mathbf{J}_i \psi=0,
\end{equation}
which is precisely Eq.~\eqref{invtensor}. Solving the quantum
constraint equation reduces $\otimes_{i=1}^n V_{j_i}$ to the invariant
tensor subspaces $\mathrm{Inv}(\otimes_{i=1}^n V_{j_i})$.

The invariant tensors parameterizes the quantum geometry of a
polyhedron with $n$ faces. A number of geometrical operators can be
defined on the Hilbert space $\mathrm{Inv}(\otimes_{i=1}^n
V_{j_i})$. For instance, from the classical face area $|\vec{A}_i|$,
we have the area operator
\begin{alignat}{5}
\hat{\mathrm{Ar}}_i\psi&{}=\sqrt{\vec{A}_i\cdot \vec{A}_i}\,\psi=8\pi \gamma \ell_P^2\sqrt{\mathbf{J}^2}\,\psi\nonumber\\
&{}=8\pi \gamma \ell_P^2\sqrt{j_i(j_i+1)}\,\psi.
\end{alignat}
Here we see that the spin $j_i$ is the quantum number of the $i$-th
face area. The discreteness of $j$ implies that the area spectrum is
discrete at the quantum level (at Planck scale). The quantum volume
operator can also be defined by quantizing the classical expression of
volume,e.g., for a tetrahedron $n=4$
\begin{equation*}
\hat{\text{Vol}}_{\text{tetrahedron}}=\frac{\sqrt{2}}{3}\sqrt{|\hat{\vec{A}}_1\cdot(\hat{\vec{A}}_2\times\hat{\vec{A}}_3)|}.
\end{equation*}
The volume operator always commutes with the area operator. The
eigenvalue problem of the volume operator can be solved in the Hilbert
space $\otimes_{i=1}^n V_{j_i}$ of invariant tensors. The operator
spectrum (eigenvalues) is again discrete (the volume spectrum is
discrete at Planck scale). The eigenstates corresponding to different
volume eigenvalues form a complete orthonormal basis of
$\otimes_{i=1}^n V_{j_i}$. Here we see that at the quantum level, the
invariant tensors in $\otimes_{i=1}^n V_{j_i}$ actually parameterize
the different quantum shapes of the polyhedron with the same face
areas $j_i$. The different (quantum) shapes of polyhedron correspond
to the different (quantum) volumes.

As an example, the trivalent invariant tensor $\ket{\psi_3}$ given in
Eq.~\eqref{eq:tri} as a perfect tensor also has an geometrical
interpretation. Namely, given a pair of triangles with the same edge
lengths $j_1,j_2,j_3$, we pick two pairs of edges from two triangles
of the same length and glue each pair. The gluing corresponds to
taking the inner product $\sum_{m_1,m_2}$ in the orthogonality
relation. Because the triangle with fixed edge lengths is rigid,
gluing two pairs of edges makes the last pair of edges congruent. This
congruence corresponds to $\delta_{m_3,m_3'}$ in
Eq.~\eqref{cong}.  From this example, we see that given the geometrical
interpretation of an invariant tensor as polygon or polyhedron, the
perfectness of the invariant tensor relates to certain rigidity of
the polygon or polyhedron.

\section{Perfect tensors as quantum error-correcting codes}

An $n$-qudit state/tensor $\ket{\psi_n}$ is perfect if
for any bipartition, whose number of particles $k$ in the smaller part
satisfies $1\leq k\leq \lfloor n/2\rfloor$, the entropy of the reduced
density matrix is maximal.  An $n$-qubit perfect tensor can be equivalently viewed
as an $[\![n,0,\delta]\!]_d$ quantum error-correcting code with the code
distance $\delta=\lfloor n/2\rfloor+1$.

For the case of $n$-qubits (i.e., $d=2$), it is known that there exist
$[\![2,0,2]\!]_2$, $[\![3,0,2]\!]_2$, $[\![5,0,3]\!]_2$ and $[\![6,0,4]\!]_2$ quantum
codes. However the $[\![n,0,\lfloor n/2\rfloor+1]\!]_2$ code does not
exist for $n=4$ and $n>6$.  Recently, it has been shown that a code
$[\![7,0,4]\!]_2$ does not exist either \cite{HGS16}.

For $d>2$ (possibly with the exception of $d=6$), $[\![4,0,3]\!]_d$ exist,
i.e., there exist perfect $4$-valent tensors. When $d>2$ is a prime
power, we can use a CSS-type code derived from a classical MDS code
with generator matrix
\begin{equation}
  G=\begin{pmatrix}
  1&0&1&1\\
  0&1&1&\alpha,
  \end{pmatrix},
\end{equation}
where $\alpha$ is an arbitrary element of the field different from $0$
and $1$. When $d=d_1d_2$ is a composite odd number or divisible by
four, we can take the tensor product of codes $[\![4,0,3]\!]_{d_1}$
and $[\![4,0,3]\!]_{d_2}$, considered as code of length four over
dimension $d$.  Finally, when $d$ is divisible by two, but not by
four, one can use the construction given in \cite{PhysRevA.92.032316}
using a pair of mutually orthogonal latin squares (MOLS) of order $d$.

\section{The derivation of Eqs. (\ref{eq:omega1}) and (\ref{eq:omega2})}
\label{app:eq}

By Lemma $2$, we can simply set
$\alpha(J)=\frac{\sqrt{2J+1}}{2j+1}\omega(J)$, where $\omega(J)$ is a
phase factor. Substituting
$\alpha(J)=\frac{\sqrt{2J+1}}{2j+1}\omega(J)$, simplify
$\tr_{13}\ket{\psi_4} \bra{\psi_4}$ and $\tr_{14}\ket{\psi_4}
\bra{\psi_4}$, then we have
\begin{alignat*}{5}
&\rho_{24}^{m_2 m_4,m_2' m_4'}\nonumber\\
&\quad{}= \frac{1}{d^2}\sum_{m_1,m_3}F_{m_1m_2m_3m_4}F^{\ast}_{m_1m_2'm_3m_4'},\nonumber\\[1ex]%\label{13a}
&\rho_{23}^{m_2 m_3,m_2' m_3'}\nonumber\\
&\quad= \frac{1}{d^2}\sum_{m_1,m_4}F_{m_1m_2m_3m_4}F^{\ast}_{m_1m_2'm_3'm_4},\nonumber%\label{14a}
\end{alignat*}
where
\begin{alignat*}{5}
F_{m_1m_2m_3m_4}&{}=\sum_{J,M}\omega(J)(-1)^{J-M}\\
&\quad\times C^{\, j\,\, j}_{m_1 m_2\, J M}C^{\, j\,\, j}_{m_3\, m_4 J-M}.
\end{alignat*}

Consider the special case $\rho_{24}^{jj,jj}$ and $\rho_{23}^{jj,jj}$,
Eq.~\eqref{perf} leads to
\begin{alignat}{5}
\rho_{24}^{jj,jj} &{}= \frac{1}{d^2}|F_{-j,j,-j,j}|^2=\frac{1}{d^2},\label{131}\\
\rho_{23}^{jj,jj} &{}= \frac{1}{d^2}|F_{-j,j,j,-j}|^2=\frac{1}{d^2},\label{141}
\end{alignat}
where
\begin{alignat*}{5}
F_{-j,j,-j,j}&{}=\sum_{J}\omega(J)(-1)^{J}C^{\, j\,\, j}_{-j\, j\, J\, 0}C^{\, j\,\, j}_{-j\, j\, J\,0},\\%\label{132}\\
F_{-j,j,j,-j}&{}=\sum_{J}\omega(J)C^{\, j\,\, j}_{-j\, j\, J\, 0}C^{\, j\,\, j}_{-j\, j\, J\,0}.%\label{142}
\end{alignat*}
By Lemma~\ref{lemma3}, we know that Eq.~\eqref{131} and Eq.~\eqref{141} contradict each other.

\section{The calculation of $\bar{Z}_1$ and $\bar{Z}_0$}
\label{app:Z1Z0}

In this section, we calculate the average of $Z_1$ and $Z_0$.

By Schur's Lemma \cite{church,Qi1}
\begin{alignat}{5}
&\overline{\rho\otimes\rho}\nonumber\\
&=\int\rmd U\, (U\otimes U)|0\rangle\langle 0|\otimes |0\rangle\langle 0|(U^\dagger\otimes U^\dagger)\nonumber\\
&=\frac{1}{\dim(\ch_{\text{inv}})^2+\dim(\ch_{\text{inv}})}\lt(\ci+\cf\rt),\notag\\
\end{alignat}
where $|0\rangle$ is an arbitrary reference state in
$\ch_{\text{inv}}$. The average is over all unitary operators $U$ on
$\ch_{\text{inv}}$; $\ci$ is the identity operator on
$\ch_{\text{inv}}\otimes\ch_{\text{inv}}$, and $\cf$ is the swap
operator
\begin{alignat}{5}
\ci |I\rangle\otimes |I'\rangle&{}=|I\rangle\otimes |I'\rangle,\nonumber\\
\cf |I\rangle\otimes |I'\rangle&{}=|I'\rangle\otimes |I\rangle.
\end{alignat}
The average $\overline{Z_1}$ is computed as follows
\begin{widetext}
\begin{alignat}{5}
&\lt[\dim(\ch_{\text{inv}})^2+\dim(\ch_{\text{inv}})\rt]\overline{Z_1}\nonumber\\
&\ {}=\sum_{\vec{m},\vec{m}'}\langle m_1;m_2;m_3;m_4|\otimes\langle m'_1;m'_2;m'_3;m'_4|\lt(\ci+\cf\rt)\cf_{34} | m_1;m_2;m_3;m_4\rangle\otimes | m'_1;m'_2;m'_3;m'_4\rangle\nonumber\\
&\ {}=\sum_{\vec{m},\vec{m}'}\langle m_1;m_2;m_3;m_4|\otimes\langle m'_1;m'_2;m'_3;m'_4|\lt(\ci+\cf\rt) | m_1;m_2;m'_3;m'_4\rangle\otimes | m'_1;m'_2;m_3;m_4\rangle.\nonumber
\end{alignat}
\end{widetext}
$\ci$ and $\cf$ act on the invariant tensors in
$\ch_{\text{inv}}\otimes\ch_{\text{inv}}$. So when they acting on
$|m_1;m_2;m'_3;m'_4\rangle\otimes | m'_1;m'_2;m_3;m_4\rangle$, they
give
\begin{widetext}
\begin{alignat}{5}
&\lt(\ci+\cf\rt) | m_1;m_2;m'_3;m'_4\rangle\otimes | m'_1;m'_2;m_3;m_4\rangle\nonumber\\
&\quad{}=\lt(\ci+\cf\rt) P_{inv}\otimes P_{inv}| m_1;m_2;m'_3;m'_4\rangle
\otimes | m'_1;m'_2;m_3;m_4\rangle\nonumber\\
&\quad{}=\sum_{I,I'}|I\rangle I^{j_1,j_2,j_3,j_4}_{m_1,m_2,m_3',m_4'}\otimes|I'\rangle I'^{j_1,j_2,j_3,j_4}_{m_1',m_2',m_3,m_4}
+\sum_{I,I'}|I\rangle I^{j_1,j_2,j_3,j_4}_{m_1',m_2',m_3,m_4}\otimes|I'\rangle I'^{j_1,j_2,j_3,j_4}_{m_1,m_2,m_3',m_4'},\nonumber
\end{alignat}
\end{widetext}
where we have used $I$ to label an orthonormal basis in
$\ch_{\text{inv}}$. $P_{\text{inv}}=\sum_{I}|I\rangle\langle I|$ is
the projector onto the invariant subspace
$\ch_{\text{inv}}$. $I^{j_1,j_2,j_3,j_4}_{m_1,m_2,m_3,m_4}=\langle
I|m_1,m_2,m_3,m_4\rangle$ is the invariant tensor component.
$\overline{Z_1}$ is thus expressed as
\begin{widetext}
\begin{alignat}{5}
&\lt[\dim(\ch_{\text{inv}})^2+\dim(\ch_{\text{inv}})\rt]\overline{Z_1}\nonumber\\
&\quad{}=\sum_{\vec{m},\vec{m}'}\Bigg(\sum_{I,I'}(I^{j_1,j_2,j_3,j_4}_{m_1,m_2,m_3,m_4})^* I^{j_1,j_2,j_3,j_4}_{m_1,m_2,m_3',m_4'}(I'^{j_1,j_2,j_3,j_4}_{m_1',m_2',m_3',m_4'})^* I'^{j_1,j_2,j_3,j_4}_{m_1',m_2',m_3,m_4}\nonumber\\
&\qquad{}+\sum_{I,I'} (I^{j_1,j_2,j_3,j_4}_{m_1,m_2,m_3,m_4})^* I^{j_1,j_2,j_3,j_4}_{m_1',m_2',m_3,m_4} (I'^{j_1,j_2,j_3,j_4}_{m_1',m_2',m_3',m_4'})^* I'^{j_1,j_2,j_3,j_4}_{m_1,m_2,m_3',m_4'}\Bigg).\label{avgZ1}
\end{alignat}
\end{widetext}
We choose the orthonormal basis $|I\rangle$ to be such that (as we did in Eq.~\eqref{4-pure})
\begin{alignat}{5}
&I^{j_1,j_2,j_3,j_4}_{m_1,m_2,m_3,m_4}\nonumber\\
&\quad{}=\sum_{M}\frac{(-1)^{I-M}}{\sqrt{2I+1}}C^{j_1,j_2,I}_{m_1,m_2,M}C^{j_3,j_4,I}_{m_3,m_4,-M}.\nonumber
\end{alignat}
It is straightforward to check the orthonormality
$\sum_{\vec{m}}(I^{j_1,j_2,j_3,j_4}_{m_1,m_2,m_3,m_4})^*\tilde{I}^{j_1,j_2,j_3,j_4}_{m_1,m_2,m_3,m_4}=\delta_{I,\tilde{I}}$. Inserting
into $\overline{Z_1}$, we find the first term in Eq.~\eqref{avgZ1}
gives
\begin{widetext}
\begin{alignat}{5}
&\sum_{I,I'}\sum_{\vec{m},\vec{m}'}(I^{j_1,j_2,j_3,j_4}_{m_1,m_2,m_3,m_4})^*I^{j_1,j_2,j_3,j_4}_{m_1,m_2,m_3',m_4'}(I'^{j_1,j_2,j_3,j_4}_{m_1',m_2',m_3',m_4'})^*I'^{j_1,j_2,j_3,j_4}_{m_1',m_2',m_3,m_4}\nonumber\\
&={}\sum_{I,I'}\sum_{\vec{m},\vec{m}'}\sum_{M,\tilde{M}}\frac{(-1)^{2I-M-\tilde{M}}}{{2I+1}}C^{j_1,j_2,I}_{m_1,m_2,M}C^{j_3,j_4,I}_{m_3,m_4,-M}C^{j_1,j_2,I}_{m_1,m_2,\tilde{M}}C^{j_3,j_4,I}_{m_3',m_4',-\tilde{M}}\nonumber\\
&\times\sum_{N,\tilde{N}}\frac{(-1)^{2I'-N-\tilde{N}}}{{2I'+1}}C^{j_1,j_2,I'}_{m_1',m_2',N}C^{j_3,j_4,I'}_{m_3',m_4',-N}C^{j_1,j_2,I'}_{m_1',m_2',\tilde{N}}C^{j_3,j_4,I'}_{m_3,m_4,-\tilde{N}}\nonumber\\
&{}=\sum_{I,I'}\sum_{M,\tilde{M}}\sum_{N,\tilde{N}}\delta_{M,\tilde{M}}\delta_{N,\tilde{N}}\delta^{I,I'}\delta_{M,\tilde{N}}\delta^{I,I'}\delta_{N,\tilde{M}}(2I+1)^{-2}\nonumber\\
&{}=\sum_{I}(2I+1)^{-1}.
\end{alignat}
\end{widetext}
The second term in Eq.~\eqref{avgZ1} gives the same result. Therefore
\begin{equation}
\overline{Z_1}=\frac{2\sum_{I}(2I+1)^{-1}}{\dim(\ch_{\text{inv}})^2+\dim(\ch_{\text{inv}})}.
\end{equation}

The average of $Z_0$ can be computed in a similar way, by removing the swap operator $\cf_{34}$
\begin{widetext}
  \begin{small}
\begin{alignat}{5}
\overline{Z_0}&{}=\frac{1}{\dim(\ch_{\text{inv}})^2+\dim(\ch_{\text{inv}})}\sum_{\vec{m},\vec{m}'}\langle\vec{m}|\otimes\langle\vec{m}'|(\ci+\cf)|\vec{m}\rangle\otimes |\vec{m}'\rangle\nonumber\\
&{}=\frac{1}{\dim(\ch_{\text{inv}})^2+\dim(\ch_{\text{inv}})}\sum_{\vec{m},\vec{m}'}%\nonumber\\
\Bigg(\sum_{I,I'}(I^{j_1,j_2,j_3,j_4}_{m_1,m_2,m_3,m_4})^* I^{j_1,j_2,j_3,j_4}_{m_1,m_2,m_3,m_4}(I'^{j_1,j_2,j_3,j_4}_{m_1',m_2',m_3',m_4'})^* I'^{j_1,j_2,j_3,j_4}_{m_1',m_2',m_3',m_4'}\nonumber\\
&\qquad\qquad{}+\sum_{I,I'} (I^{j_1,j_2,j_3,j_4}_{m_1,m_2,m_3,m_4})^* I^{j_1,j_2,j_3,j_4}_{m_1',m_2',m_3',m_4'} (I'^{j_1,j_2,j_3,j_4}_{m_1',m_2',m_3',m_4'})^* I'^{j_1,j_2,j_3,j_4}_{m_1,m_2,m_3,m_4}\Bigg)\nonumber\\
&=1.
\end{alignat}
  \end{small}
\end{widetext}

\section{Bound on Fluctuations of $S_2$}
\label{sec:bound}
In this section, we estimate the bound on fluctuation of the second
Renyi entropy $S_2$ around the average $\overline{S_2}$, under the
asymptotical limit $j\to\infty$. Using the bound, we also show that
$S_2$ concentrates at $\overline{S_2}$ with a high probability, which
is close to $1$ as $j\to\infty$. The idea of derivation is similar to
\cite{Qi1}.

We consider the fluctuation:
\begin{equation}
\frac{\overline{(Z_{1}-\overline{Z_{1}})^2}}{\overline{Z_1}^2}=\frac{\overline{Z_{1}^2}}{\overline{Z_1}^2}-1.
\end{equation}
We compute the general average $\overline{Z_{1}^2}=\overline{\trace\lt[(\rho\otimes\rho)\cf_{34}\rt]^2}$ by using the following formula \cite{Qi1,church}
\begin{equation}
\overline{\rho^{\otimes 4}}=\frac{1}{\cc_{4}}\sum_{\sig\in\mathrm{Sym}_{4}} \sig\in\ch_{\text{inv}}^{\otimes 4}\otimes \ch_{\text{inv}}^{*\otimes 4},
\end{equation}
where
$\cc_m=(\dim\ch_{\text{inv}}+m-1)!/(\dim\ch_{\text{inv}}-1)!$. The sum
is over all permutations $\sig$ acting on
$\ch_{\text{inv}}^{4}$. Inserting the above formula, we have
\begin{widetext}
\begin{equation}
\overline{Z_{1}^2}=\frac{1}{\cc_{4}}\sum_{\sig\in\mathrm{Sym}_{4}}\sum_{\vec{m}^{(i)}}\langle\vec{m}^{(1)}|\cdots\langle\vec{m}^{(4)}|\sig\cf_{34}^{\otimes 2}|\vec{m}^{(1)}\rangle\cdots|\vec{m}^{(4)}\rangle.
\end{equation}
The operation of $\sig\cf_{34}^{\otimes 2}$ gives
\begin{alignat}{5}
&\sig\bigotimes_{i=1,3}\cf_{34}|{m}_1^{(i)};{m}_2^{(i)};{m}_3^{(i)};{m}_4^{(i)}\rangle|{m}_1^{(i+1)};{m}_2^{(i+1)};{m}_3^{(i+1)};{m}_4^{(i+1)}\rangle\nonumber\\
&{}=\sum_{I^{(i)}}\sig\bigotimes_{i=1,3}|I^{(i)}\rangle\langle I^{(i)}|{m}_1^{(i)};{m}_2^{(i)};{m}_3^{(i+1)};{m}_4^{(i+1)}\rangle\otimes|I^{(i+1)}\rangle\langle I^{(i+1)}|{m}_1^{(i+1)};{m}_2^{(i+1)};{m}_3^{(i)};{m}_4^{(i)}\rangle\nonumber\\
&{}=\sum_{I^{(i)}}\sig\bigotimes_{i=1,3}|I^{(i)}\rangle I^{(i)}_{{m}_1^{(i)}{m}_2^{(i)}{m}_3^{(i+1)}{m}_4^{(i+1)}}\otimes|I^{(i+1)}\rangle I^{(i+1)}_{{m}_1^{(i+1)}{m}_2^{(i+1)}{m}_3^{(i)}{m}_4^{(i)}}\nonumber\\
&{}=\sum_{I^{(i)}}\bigotimes_{i=1,3}|I^{\sig(i)}\rangle \otimes|I^{\sig(i+1)}\rangle \prod^{2N}_{i\,\text{even}}I^{(i)}_{{m}_1^{(i)}{m}_2^{(i)}{m}_3^{(i+1)}{m}_4^{(i+1)}} I^{(i+1)}_{{m}_1^{(i+1)}{m}_2^{(i+1)}{m}_3^{(i)}{m}_4^{(i)}}.
\end{alignat}
Taking the inner product, we obtain
\begin{alignat}{5}
&\overline{Z_{1}^2}=\frac{1}{\cc_{4}}\sum_{\sig\in\mathrm{Sym}_{4}}\sum_{\vec{m}^{(i)}}\sum_{I^{(i)}}\nonumber\\
&\prod_{i=1,3}I^{*\ \sig(i)}_{{m}_1^{(i)}{m}_2^{(i)}{m}_3^{(i)}{m}_4^{(i)}} I^{*\ \sig(i+1)}_{{m}_1^{(i+1)}{m}_2^{(i+1)}{m}_3^{(i+1)}{m}_4^{(i+1)}} I^{(i)}_{{m}_1^{(i)}{m}_2^{(i)}{m}_3^{(i+1)}{m}_4^{(i+1)}} I^{(i+1)}_{{m}_1^{(i+1)}{m}_2^{(i+1)}{m}_3^{(i)}{m}_4^{(i)}}.
\end{alignat}
Summation over $\vec{m}^{{i}}$ yeilds
\begin{alignat}{5}
&\sum_{\vec{m}^{(i)},\vec{m}^{(i+1)}}I^{*\ \sig(i)}_{{m}_1^{(i)}{m}_2^{(i)}{m}_3^{(i)}{m}_4^{(i)}} I^{*\ \sig(i+1)}_{{m}_1^{(i+1)}{m}_2^{(i+1)}{m}_3^{(i+1)}{m}_4^{(i+1)}} I^{(i)}_{{m}_1^{(i)}{m}_2^{(i)}{m}_3^{(i+1)}{m}_4^{(i+1)}} I^{(i+1)}_{{m}_1^{(i+1)}{m}_2^{(i+1)}{m}_3^{(i)}{m}_4^{(i)}}\nonumber\\
&\quad{}=\sum_{\vec{m}^{(i)},\vec{m}^{(i+1)}}\sum_{M^{\sig(i)}}\frac{(-1)^{I^{\sig(i)}-M^{\sig(i)}}}{\sqrt{2I^{\sig(i)}+1}}C^{j_1,j_2,I^{\sig(i)}}_{m_1^{(i)}m_2^{(i)}M^{\sig(i)}}C^{j_3,j_4,I^{\sig(i)}}_{m_3^{(i)}m_4^{(i)}-M^{\sig(i)}}\nonumber\\
&\qquad\sum_{M^{(i)}}\frac{(-1)^{I^{(i)}-M^{(i)}}}{\sqrt{2I^{(i)}+1}}C^{j_1,j_2,I^{(i)}}_{m_1^{(i)}m_2^{(i)}M^{(i)}}C^{j_3,j_4,I^{(i)}}_{m_3^{(i+1)}m_4^{(i+1)}-M^{(i)}}\nonumber\\
&\qquad\sum_{M^{\sig(i+1)}}\frac{(-1)^{I^{\sig(i+1)}-M^{\sig(i+1)}}}{\sqrt{2I^{\sig(i+1)}+1}}C^{j_1,j_2,I^{\sig(i+1)}}_{m_1^{(i+1)}m_2^{(i+1)}M^{\sig(i+1)}}C^{j_3,j_4,I^{\sig(i+1)}}_{m_3^{(i+1)}m_4^{(i+1)}-M^{\sig(i+1)}}\nonumber\\
&\qquad\sum_{M^{(i+1)}}\frac{(-1)^{I^{(i+1)}-M^{(i+1)}}}{\sqrt{2I^{(i+1)}+1}}C^{j_1,j_2,I^{(i+1)}}_{m_1^{(i+1)},m_2^{(i+1)},M^{(i+1)}}C^{j_3,j_4,I^{(i+1)}}_{m_3^{(i)}m_4^{(i)}-M^{(i+1)}}\nonumber\\
&{}=\sum_{M^{\sig(i)},M^{(i)},M^{\sig(i+1)},M^{(i+1)}}
\frac{1}{({2I^{(i)}+1})({2I^{(i+1)}+1})}\delta^{I^{\sig(i)}I^{(i)}}\delta^{I^{\sig(i+1)}I^{(i+1)}}
\delta^{I^{\sig(i)}I^{(i+1)}}\delta^{I^{(i)}I^{\sig(i+1)}}\nonumber\\
&\qquad\qquad\times\delta_{M^{\sig(i)}M^{(i)}}\delta_{M^{\sig(i+1)}M^{(i+1)}}\delta_{M^{\sig(i)}M^{(i+1)}}\delta_{M^{(i)}M^{\sig(i+1)}}\nonumber\\
&{}=
\frac{1}{({2I^{(i)}+1})^2}\delta^{I^{\sig(i)}I^{(i)}}\delta^{I^{\sig(i+1)}I^{(i+1)}}
\delta^{I^{(i)}I^{(i+1)}}\delta^{I^{(i)}I^{(i+1)}}
\sum_{M^{(i)},M^{(i+1)}}\delta_{M^{(i)}M^{(i+1)}}\delta_{M^{(i)}M^{(i+1)}}\nonumber\\
&{}=\frac{1}{({2I^{(i)}+1})}\delta^{I^{\sig(i)}I^{(i)}}\delta^{I^{\sig(i+1)}I^{(i+1)}}
\delta^{I^{(i)}I^{(i+1)}}\delta^{I^{(i)}I^{(i+1)}}.
\end{alignat}
Inserting the result into $\overline{Z_{1}^2}$ gives
\begin{alignat}{5}
\overline{Z_{1}^2}&{}=\frac{1}{\cc_{4}}\sum_{\sig\in\mathrm{Sym}_{4}}\sum_{I^{(i)}}\prod_{i=1,3}\frac{1}{({2I^{(i)}+1})}\delta^{I^{\sig(i)}I^{(i)}}\delta^{I^{\sig(i+1)}I^{(i+1)}}
\delta^{I^{(i)}I^{(i+1)}}\nonumber\\
&{}=\frac{1}{\cc_{4}}\lt[4\lt(\sum_{I=0}^{2j}\frac{1}{2I+1}\rt)^2+3!\sum_{I=0}^{2j}\frac{1}{(2I+1)^2}\rt].
\end{alignat}
\end{widetext}

We have set $j_1=j_2=j_3=j_4=j$. As $j\to\infty$,
$\sum_{I=0}^{2j}\frac{1}{2I+1}$ is divergent, while
$\sum_{I=0}^{\infty}\frac{1}{(2I+1)^2}=\frac{\pi^2}{8}$, and
$\sum_{I=0}^{2j}\frac{1}{(2I+1)^{m\geq2}}$ are all
convergent. Therefore if we denote by
$\L\equiv\sum_{I=0}^{2j}\frac{1}{2I+1}$,
\begin{equation}
\frac{\overline{Z_{1}^2}}{\overline{Z_{1}}^2}=\frac{(\cc_2)^2}{\cc_{4}}\lt[1+\frac{3\pi^2}{16}\lt(\frac{1}{\L}\rt)^2\rt].
\end{equation}
Given that $\frac{(\cc_2)^2}{\cc_{4}}< 1$,
\begin{equation}
\overline{\lt(\frac{Z_{1}}{\overline{Z_1}}-1\rt)^2}=\frac{\overline{Z_{1}^2}}{\overline{Z_{1}}^2}-1<\frac{3\pi^2}{16}\lt(\frac{1}{\L}\rt)^2
\end{equation}
By Markov's inequality,
\begin{alignat}{5}
&\mathrm{Prob}\lt(\lt|\frac{Z_{1}}{\overline{Z_1}}-1\rt|\geq\frac{\delta}{4}\rt)\nonumber\\
&\leq \frac{\overline{\lt(\frac{Z_{1}}{\overline{Z_1}}-1\rt)^2}}{\lt(\frac{\delta}{4}\rt)^2}<\frac{3\pi^2}{\delta^2\L^2}.\label{prob1}
\end{alignat}

On the other hand, we can also show that
\begin{widetext}
\begin{alignat}{5}
\overline{Z_{0}^2}&{}=\frac{1}{\cc_{4}}\sum_{\sig\in\mathrm{Sym}_{4}}\sum_{I^{(i)}}\delta^{I^{\sig(1)},I^{(1)}}\delta^{I^{\sig(2)},I^{(2)}}\delta^{I^{\sig(3)},I^{(3)}}\delta^{I^{\sig(4)},I^{(4)}}\nonumber\\
&{}\simeq \frac{1}{\cc_{4}}\lt[(2j+1)^4+6(2j+1)^3+10(2j+1)^2+7(2j+1)\rt].
\end{alignat}
\end{widetext}

Therefore as $j\to\infty$
\begin{equation}
\overline{\lt(\frac{Z_{0}}{\overline{Z_0}}-1\rt)^2}<\frac{3}{j}+O(j^{-2})
\end{equation}
and
\begin{alignat}{5}
&\mathrm{Prob}\lt(\lt|\frac{Z_{0}}{\overline{Z_0}}-1\rt|\geq\frac{\delta}{4}\rt)\nonumber\\
&<\frac{16}{\delta^2}\lt[\frac{3}{j}+O(j^{-2})\rt]<\frac{3\pi^2}{\delta^2\L^2}.\label{prob2}
\end{alignat}

The bounds Eq.~\eqref{prob1} and \eqref{prob2} imply that with the probably of at least $1-\frac{3\pi^2}{\delta^2\L^2}$, we have $\lt|\frac{Z_{0,1}}{\overline{Z_0,1}}-1\rt|\leq\frac{\delta}{4}$. Then we have
\begin{alignat}{5}
\lt|S_2-\overline{S_2}\rt|&{}=\lt|\ln\frac{{Z_1}}{{Z_0}}-\ln\frac{\overline{Z_1}}{\overline{Z_0}}\rt|\nonumber\\
&{}=\lt|\ln\frac{{Z_1}}{\overline{Z_1}}-\ln\frac{{Z_0}}{\overline{Z_0}}\rt|\nonumber\\
&{}\leq\lt|\ln\frac{{Z_1}}{\overline{Z_1}}\rt|+\lt|\ln\frac{{Z_0}}{\overline{Z_0}}\rt|\nonumber\\
&{}\leq\frac{\delta}{2}+\frac{\delta}{2}=\delta,
\end{alignat}
where we have used that for $\delta\leq 2$, $|\ln(1\pm \delta/4)|\leq \delta/2$.

Therefore we have shown that for any small $\delta>0$, there is a large probability
\begin{equation}
P_\delta=1-\frac{3\pi^2}{\delta^2\L^2},
\end{equation}
which is close to 1 as $j\to\infty$ (since
$\L=\sum_{I=0}^{2j}\frac{1}{2I+1}\to\infty$), such that
$\lt|S_2-\overline{S_2}\rt|\leq\delta$, i.e., the second Renyi entropy
is close to the average value $\overline{S_2}\sim\ln\lt[(2j+1)^2\rt]$.

\end{document}